\documentclass[envcountsect]{llncs}

\pdfoutput=1

\usepackage{amssymb,amsmath}

\usepackage{ucs}
\usepackage[utf8x]{inputenc}
\usepackage{url}

\usepackage[noend]{algorithmic}
\usepackage{algorithm}
\usepackage{graphicx}
\usepackage{verbatim}
\usepackage{proof}
\usepackage[usenames,dvipsnames]{color}

\usepackage{booktabs}
\graphicspath{{img/}}

\newcommand{\<}{\langle}
\renewcommand{\>}{\rangle}
\newcommand{\agents}{{\cal A}}

\newcommand{\bigO}[1]{{\cal O}{#1}}
\newcommand{\likes}[2]{l(j,i)}

\newcommand{\lstrat}[1]{\lambda_{#1}}
\newcommand{\strat}{s}
\newcommand{\strats}[1]{strat(#1)}

\newcommand{\vote}[2]{V_{#1}(#2)}
\newcommand{\votep}[1]{P(#1)}
\newcommand{\kvote}[3]{V_{#1}(#2,#3)}
\newcommand{\kvotep}[2]{P(#1,#2)}

\newcommand{\Nat}{{\mathbb N}}

\newcommand{\acro}[1]{\textsc{#1}}
\newcommand{\act}{\mathbb{A}}

\newcommand{\catld}[1]{\langle\!\langle #1\rangle\!\rangle}

\renewcommand{\iff}{\ \text{iff}\ }
\newcommand\ie{i.e{.\ }}
\newcommand\eg{e.g{.\ }}

\newcommand{\roles}{\mathcal{R}}

\newcommand{\exptime}{$\mathsf{EXPTIME}$}

\newcommand{\np}{$\mathsf{NP}$}

\title{No big deal: introducing roles to reduce the size of ATL
  models\thanks{Accepted for presentation at LAMAS 2012
    workshop on June 5, 2012 in Valencia, Spain.}}

\author{Sjur Dyrkolbotn\inst{1} \and Piotr Kaźmierczak\inst{2,3}\thanks{Piotr
    Kaźmierczak's research was supported by the Research Council of Norway project 194521 (FORMGRID).} \and
  Erik Parmann\inst{1} \and Truls Pedersen\inst{3}}
\institute{Department of Informatics,
  University of Bergen \and Department of Computing,
  Mathematics and Physics, Bergen University College \and Department
  of Information Science and Media Studies, University of Bergen
  \email{sjur.dyrkolbotn@uib.no, phk@hib.no,
    erik.parmann@uib.no, truls.pedersen@uib.no}}
\begin{document}
\maketitle

\begin{abstract}
  In the following paper we present a~new semantics for the well-known
  strategic logic \acro{atl}. It is based on adding \emph{roles} to
  concurrent game structures, that is at every state, each agent
  belongs to exactly one role, and the role specifies what actions are
  available to him at that state. We show advantages of the new
  semantics, analyze model checking complexity and prove equivalence
  between standard \acro{atl} semantics and our new approach.
\end{abstract}
 
\section{Introduction}

One of the most intensively studied
\cite{JamHoe04-0,KacPen0411-0,HoeWoo0505-0} areas of research in the
field of multi-agent systems are \emph{strategic} or
\emph{cooperation} logics -- formalisms that allow for reasoning about
agents' strategies and behavior in a~multi-agent setting. Two of the
most known logics are Marc Pauly's Coalition Logic (\acro{cl})
\cite{Pau01-0,Pau0202-0} and Alur, Henzinger and Kupferman's
Alternating-time Temporal Logic (\acro{atl})
\cite{alur2002alternating}, which can be considered a~temporal
extension of Coalition Logic. Both these logics gained much popularity
and generated a~`zoo' of derivatives
\cite{HoeWoo03-0,HoeJamWoo05-0,AgoHoeWoo08-0,AleLogNga11-0,AgoHoeWoo09-1}.

This popularity is in no small part due to relative high expressive
power of both \acro{cl} and \acro{atl}, but also due to low complexity
of model checking problems for these respective logics. Model checking
of Coalition Logic can be solved in polynomial time in the size of the
model and the length of the formula \cite{Pau01-0}. It remains
polynomial for \acro{atl} as well \cite{alur2002alternating}, which is
considered a~very good result. However, as investigated by Jamroga and
Dix \cite{JamDix05-0}, in both cases the number of agents must be
\emph{fixed}. If it is not then model checking of \acro{atl} models
represented as \emph{alternating transition systems} is \np-complete,
and if the models are represented as \emph{concurrent game structures}
(\acro{cgs}) it becomes $\mathsf{\Sigma^{P}_{2}}$-complete. Also, van
der Hoek, Lomuscio and Wooldridge show \cite{van2006complexity} that
complexity of model checking for \acro{atl} is sensitive to model
representation. It is polynomial only if an \emph{explicit}
enumeration of \emph{all} components of the model is assumed.  For
models represented in a~\emph{(simplified) reactive modules language}
(\acro{rml}) \cite{AluHen9907-0} complexity of model checking for
\acro{atl} becomes as hard as the satisfiability problem for this
logic, namely \exptime~\cite{van2006complexity}.

We present an alternative semantics that interprets formulas of
ordinary \acro{atl} over concurrent game structures with
\emph{roles}. As we describe in Section \ref{sec:rcgs}, such
structures introduce an extra element -- a set $R$ of roles. Agents
belonging to the same role are considered \emph{homogeneous} in the
sense that all consequences of their actions are captured by
considering only the number of \emph{votes} an action gets (one vote
per agent). We give some examples that motivate our approach and prove
equivalence with \acro{atl} based on concurrent game structures. We
then discuss model checking, showing it to be of polynomial complexity
in the size of models. This seems significant, since as long as the
number of roles remain fixed, the size of our models does \emph{not}
grow exponentially in the number of players.

The structure of our paper is as follows. We present a~revised
formalism for \acro{atl} in Section \ref{sec:ratl}, prove equivalence
with the standard one in Section \ref{sec:equiv}, discuss model
checking results in Section \ref{sec:mc} and conclude in Section
\ref{sec:concl}.

\section{Role-based semantics for ATL}
\label{sec:ratl}

The language of ordinary \acro{atl} is the following, as presented in
\cite{alur2002alternating}:

$$ \phi ::= p \mid \neg \phi \mid \phi \wedge \phi \mid\catld{A}\bigcirc\phi \ \mid \ \catld{A}\Box\phi \ \mid \catld{A}\phi\mathcal U\phi$$
where $p$ is propositional letter, and $A$ is a~coalition of
agents. We follow standard abbreviations (\eg $\catld{~}$ for
$\catld{\emptyset}$) and skip connectives that are derivable.

\subsection{Concurrent Game Structures with Roles}
\label{sec:rcgs}

In this section we will introduce \emph{concurrent game structures
  with roles} (\acro{rcgs}) and consider some examples. We will be
using the notation $[n] = \{1,\ldots,n\}$, and we will let $A^B$
denote the set of functions from $B$ to $A$. We will often work with
tuples $v = \<v_1,\ldots,v_n\>$ and we will often view $v$ as a
function with domain $[n]$ and write $v(i)$ for $v_i$. We will do
addition and subtraction on tuples of the same arity component by
component, e.g. for $v = \<v_1,\ldots,v_n\>, v' =
\<v'_1,\ldots,v'_n\>$, $v-v' = \<v_1-v'_1,\ldots,v_n-v'_n\>$. Given a
function $f: A \times B \to C$ and $a \in A$, we will use $f_a$ to
denote the function $B \to C$ defined by $f_a(b) = f(a,b)$ for all $b
\in B$.

\begin{definition}
  An \acro{rcgs} is a tuple $H = \langle \agents, R, \roles, Q, \Pi,
  \pi, \act, \delta\rangle$ where:
  \begin{itemize}
  \item $\agents$ is a non-empty set of players. In this text we
    assume $\agents = [n]$ for some $n \in \Nat$, and we will reserve
    $n$ to mean the number of agents.
  \item $Q$ is the non-empty set of states.
  \item $R$ is a non-empty set of roles. In this text we assume $R
    =[i]$ for some $i \in \Nat$.
  \item $\roles: Q \times R \to \wp(\agents)$, such that for every $q
    \in Q$ we have
    \begin{itemize}
    \item For all $r,r' \in R$, if $r \neq r'$ then $\roles(q,r) \cap
      \roles(q,r') = \emptyset$
    \item $\bigcup_{r \in R} \roles(q,r) = \agents$
    \end{itemize}
    For a coalition $A \subseteq \agents$ we write $A_{r,q}$ for the
    agents in $A$ which belong to role $r$ at $q$, i.e. $A_{r,q} =
    \mathcal R(q,r) \cap A.$
  \item $\Pi$ is a set of propositional letters and $\pi: Q \to
    \wp(\Pi)$ maps states to the propositions true at that state.
  \item $\act : Q \times R \to \mathbb N^+$ is the number of available
    actions in a given state for a given role.
  \item For $\mathcal{A}=[n] = \{1,\ldots n\}$, we say that the set of
    \emph{complete} votes for a role $r$ in a state $q$ is $\vote r q
    = \{v_{r,q} \in [n]^{[\act(q,r)]} \mid \sum_{1 \leq a \leq
      \act(q,r)} v_{r,q}(a) = |\mathcal R(q,r)|\}$, the set of
    functions from the available actions to the number of agents
    performing the action. The functions in this set account for the
    actions of \emph{all} the agents. The set of \emph{complete}
    profiles at $q$ is $\votep q = \prod_{r \in R}\vote r q$. For each
    $q \in Q$ we have a transition function at $q$, $\delta_q: \votep
    q \to Q$ defining a partial function $\delta: Q \times \bigcup_{q
      \in Q}\votep q \to Q$ such that for all $q \in Q$, $P \in \votep
    q$, $\delta(q,P) = \delta_q(P)$
  \end{itemize}
\end{definition}
To illustrate how \acro{rcgs} differs from an ordinary concurrent game
structure, we provide some examples.

\begin{example}
  We construct an example similar to the well-known train-controller
  scenario \cite{alur2002alternating}, but in contrast to the
  original, in our scenario there are $n_t$ trains. Consider
  a~turn-based synchronous game structure with roles $S_{train} =
  \langle \agents, $ $R, \roles, Q, \Pi, \pi, \act, \delta\rangle$
  where:
  \begin{itemize}
  \item $\agents = \{1,\dots,n_t, n_t+1\}$. There are $n_t$ trains and
    one controller.
  \item $R = \{train, ctr\}$. There are two roles: one for trains and
    one for the controller.
  \item $Q = \{q_{0},q_{1},q_{2},q_{3}\}$.
  \item $\mathcal R(q_i, train) = [n_t]$, and $\mathcal R(q_i, ctr) =
    \{n_{t} +1\}$, for all $q_i \in Q$.
  \item $\Pi = \{out\_of\_gate,in\_gate,request,grant\}$
  \item $\pi(q_{0}) = \{out\_of\_gate\}$, $\pi(q_{1}) =
    \{out\_of\_gate,request\}$, \\$\pi(q_{2}) =
    \{out\_of\_gate,grant\}$, $\pi(q_{3}) = \{in\_gate\}$.
  \item $ \act(q_0,train) = 2, \quad \act(q_0, ctr) = 1, \quad \act(q_1,train) = 1, \quad \act(q_1,ctr)= 3 $, \\
    $ \act(q_2,train) = 2, \quad \act(q_2, ctr) = 1, \quad
    \act(q_3,train) = 1, \quad \act(q_3, ctr) = 2 $.
  \item and finally
    \begin{align*}
      &\delta(q_0, \<(0, n_t), 1\>) =\delta(q_1, \<n_t, (1, 0, 0)\>) = \delta(q_2, \<(0, n_t), 1\>) \\
      &\quad\quad =\delta(q_3, \<(a, n_t-a), 1\>) = q_0 & \text{ where } 1 \leq a \leq n_{t}\\
      &\delta(q_0, \<(a, n_t-a), 1\>) = \delta(q_1, \<n_t, ((0,1,0))\>) = q_1 & \text{ where } 1 \leq a \leq n_{t}\\
      &\delta(q_1, \<n_t, (0,0,1)\>) = \delta(q_2, \<(a, n_t - a), 1\>) = q_2 & \text{ where } 2 \leq a \leq n_{t}\\
      &\delta(q_2, \<(1, n_t-1), 1\>) = \delta(q_3, \<(0, n_t), 1\>) = q_3 \\
    \end{align*}

  \end{itemize}
  Figure \ref{fig:trains-ratl} presents the example in a~visual
  way. The model can be seen as a generalization of the classical
  train-controller example. In $q_0$ we stay in $q_0$ unless at least
  one train issues a request. In $q_1$ the controller behaves as
  before; it can postpone making a decision (staying in $q_1$), reject
  all requests (going to $q_0$), or accept the requests (going to
  $q_2$). In $q_2$ the trains can choose to enter the tunnel, but only
  one of them may do so; if nobody attempts to enter the grant is
  revoked (or relinquished), if more than one train attempts to enter
  we stay in $q_2$, and finally if (the trains reach an agreement and)
  only one train enters we go to $q_3$. In $q_3$ \emph{any} train may
  decide that the train in the tunnel has to leave (returning to
  $q_0$), and the train in the tunnel \emph{must} comply. This
  reflects the homogeneity among players in the trains role. The
  action of deciding to leave the tunnel is shared among all trains,
  and the train actually in the tunnel remains unidentified.

  \begin{figure}[h]
    \begin{center}
      \includegraphics[scale=.8]{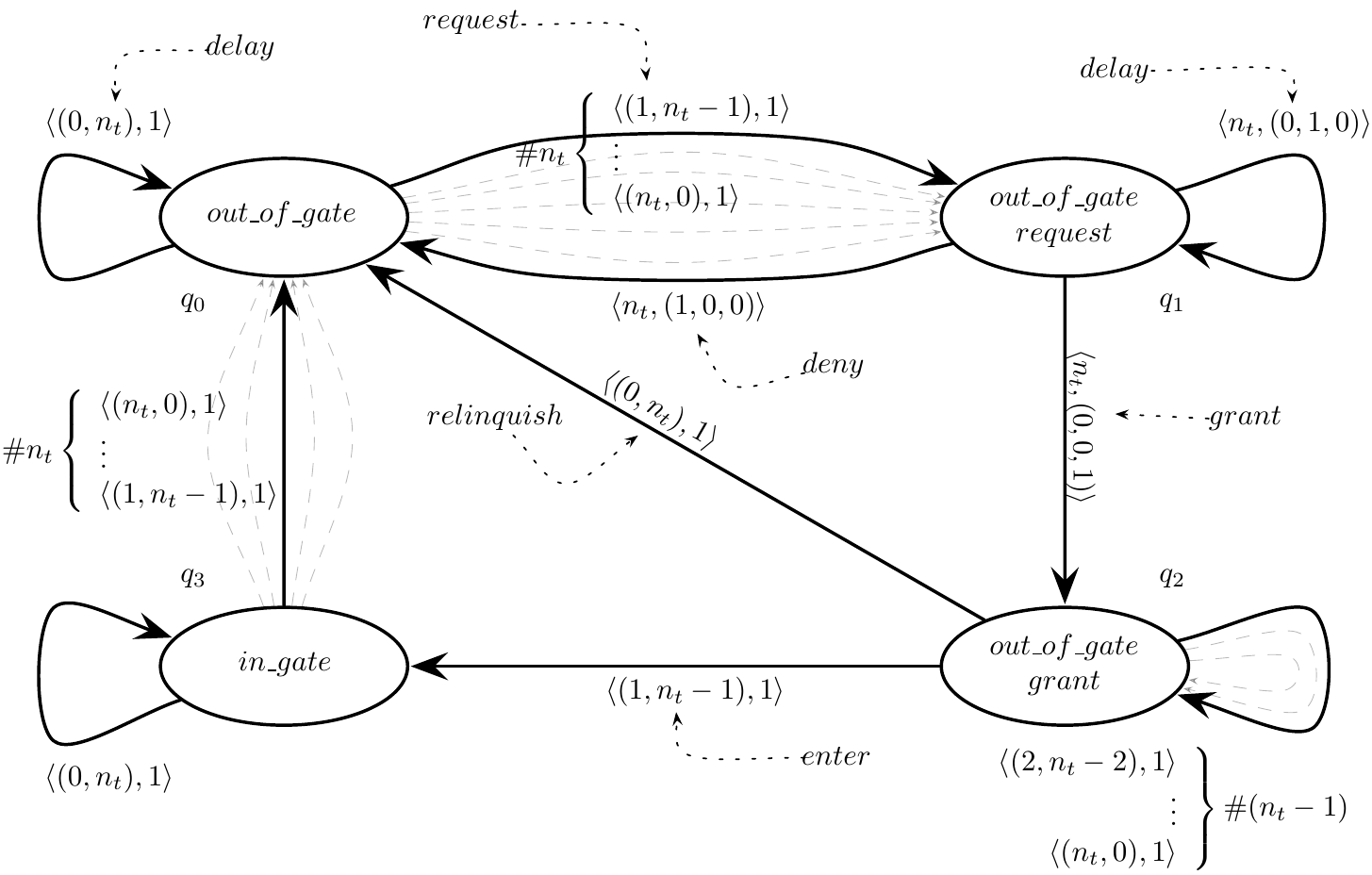}
    \end{center}
    \caption{Train controller model for $n_{t}$ trains (similar to the
      one presented in \cite{alur2002alternating}).}
    \label{fig:trains-ratl}
  \end{figure}

  Notice that in the single-train case ($n_t = 1$), the train can not
  wait before entering the tunnel after being granted permission
  (\emph{and} retain the permission). This could of course easily be
  avoided by adding another action. More importantly, in the case of
  several trains, the controller can not distinguish between the
  different trains, so permission must be granted to all or none. This
  is a consequence of the strict homogeneity in the model: not only
  are the agents homogeneous in terms of the actions available to
  them, we can not reasonably distinguish between them as long as they
  remain in the same role. Notice that this feature allow us to add
  any number of trains to the scenario without incurring more than a
  linear increase in the size of the model (total number of
  profiles). This would not be possible if we did not have roles. If
  the model above was to be rendered as a concurrent game structure,
  the number of possible ways in which trains could act would be
  exponential in all states where trains have to make a choice of what
  action to perform. This would be the case even if, as in the
  scenario above, almost all possible combinations of choices should
  be treated in the same way by the system.
\end{example}

Sometimes homogeneity is desirable. In our trains and controller
example, for instance, homogeneity strongly \emph{encourages}
cooperation among trains; no one can enter the gate unless everyone
agree, and everyone knows that whoever gets to enter \emph{must} leave
as soon as he is asked to. On the other hand, we notice that it is
impossible for any train to enter the gate unless all trains
cooperate. This might be overly restrictive. By adding more roles,
however, we can amend this while still retaining many of the benefits
of using roles.

\begin{example}
  \label{ex:example2}
  In the previous example all trains were equal before the controller;
  the controller could not distinguish between trains. We could grant
  the agents much more individual identity by simply adding one more
  role, and in this example we sketch the result of doing so.  First
  we make $n_t$ ``copies'' of the previous model sharing the state
  $q_0$. In Figure \ref{fig:autonomous-trains} we illustrate the
  resulting model for $n_t = 3$. In $q_0$ we let the trains vote for
  which train should be allowed to request permission to enter the
  tunnel. We assume majority voting, but we do not resolve ties.  It
  means that if one train, $x$, gets more votes then all others we go
  to "his" state, $q_1x$. Otherwise we just loop on $q_0$. If we get
  to a $q_1$-state, the controller can grant or reject the
  request. Contrary to the previous example the controller now knows
  which train is being proposed. If the controller grants the request,
  the selected train is put in a privileged role and given the sole
  choice of what to do with the permission.

  \begin{figure}[h]
    \includegraphics[scale=.6]{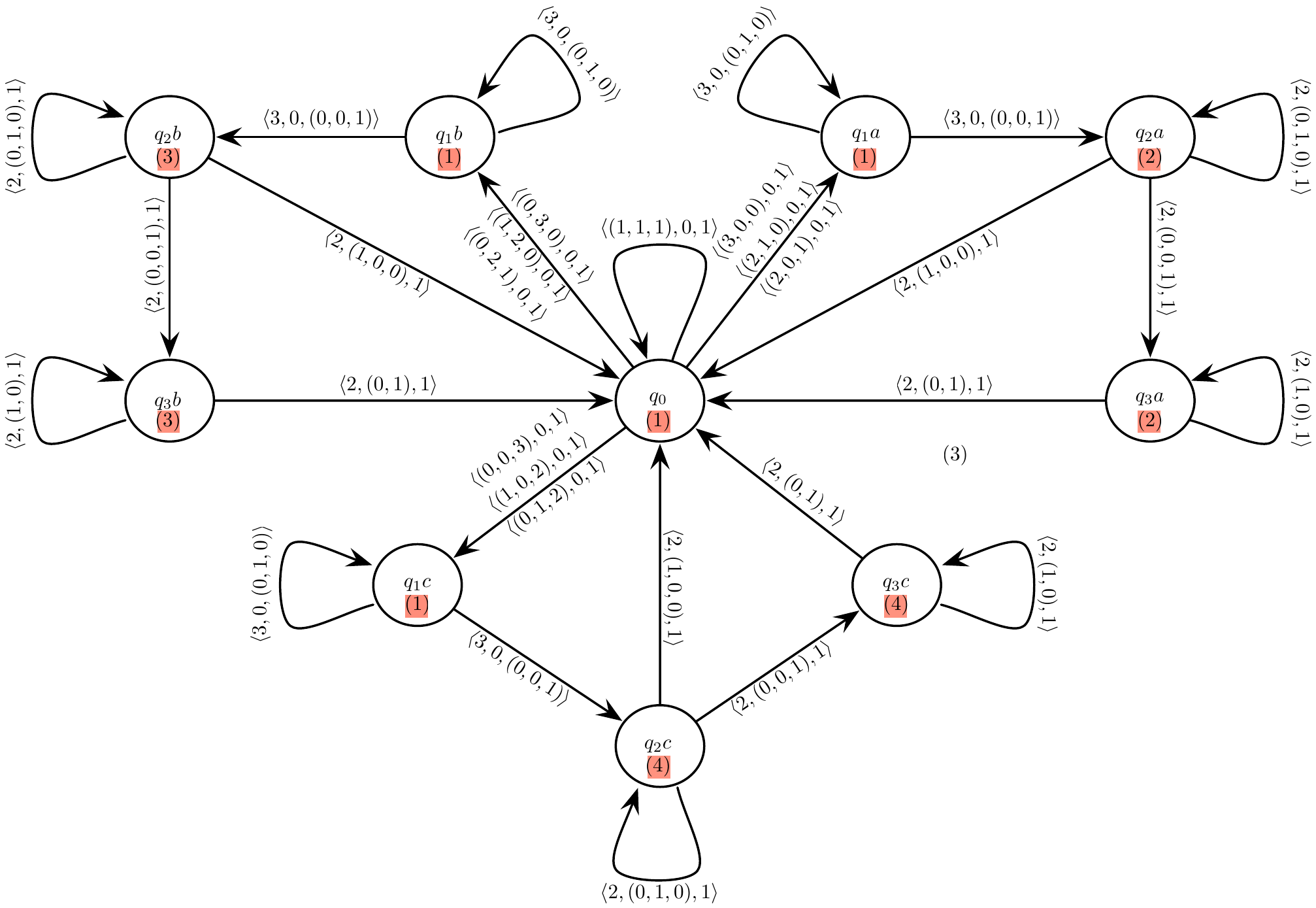}
    \caption{The ``autonomous trains'' model for $n_{t} = 3$. Numbers
      in red indicate membership in roles: $(1)\ = \<\agents,
      \emptyset, \{ctr\}\>$, $(2)\ = \<\agents\setminus\{a\},
      \{a\},\{ctr\}\>$, $(3)\ = \<\agents\setminus\{b\},
      \{b\},\{ctr\}\>$, $(4)\ = \<\agents\setminus\{c\}, \{c\},
      \{ctr\}\>$. Also, each $q_{1}$ state has a~transition pointing
      at $q_{0}$ labeled $\<3,0,(1,0,0)\>$ that was omitted from the
      picture for the sake of clarity.}
    \label{fig:autonomous-trains}
  \end{figure}

  The model has grown, so the trains gain autonomy at a cost. Still,
  this cost is much less than the cost of modelling this scenario in a
  \acro{cgs}. There, if each train is to have the option to ``vote''
  for any train in $q_0$, each train must have $n_t$ actions
  available. We would get $n_t^{n_t}$ edges leading out from $q_0$! In
  the \acro{rcgs} model we get a substantially smaller degree, the
  following table summarizes the difference (formulas for counting the
  degree are explained and discussed further in section \ref{sec:mc})
  \begin{center}
    \begin{tabular}{r l l l l l c}
      \toprule
      $n_{t}$: & 3 & 4 & 5 & 6 & \ldots & n \\ \midrule
      \acro{cgs}: & 27 & 256 & 3125 & 46656 & \ldots & $n^n$ \\ %\hline
      \acro{rcgs}: & 10 & 35 & 127 & 462 & \ldots & $\frac{(2n - 1)!}{n!(n-1)!}$ \\ \bottomrule
    \end{tabular}
  \end{center}

\end{example}

Before we move on we introduce some more notation.  Given a role $r
\in R$, a state $q$ and a coalition $A$, the set of $A$-votes for $r$
at $q$ is $\kvote r q A$, defined as follows:
$$
\kvote r q A = \left\lbrace v \in [|A_{r,q}|]^{[\act(q,r)]} ~\left|~
    \sum_{a \in [\act(q,r)]} v(a) = |A_{r,q}|\right.\right\rbrace
$$
The $A$-votes for $r$ at $q$ gives the possible ways agents in $A$
that are in role $r$ at $q$ can vote. Given a state $q$ and a
coalition $A$, we define the set of $A$-profiles at $q$:
$$
\kvotep q A = \{\langle v_1, \ldots , v_{|R|}\rangle \mid 1 \leq i
\leq |R| : v_i \in \kvote r q A\}
$$
When we say that a function $v: [\act(q,r)] \to [n]$ is a complete
vote (for $r$ at $q$), we mean that $v \in \kvote r q {\agents}$.
\medskip For any $v \in \kvote r q A$ and $w \in \kvote r q B$ we
write $v \leq w$ iff for all $i \in [\act(q,r)]$ we have $v(i) \leq
w(i)$. If $v \leq w$, we say that $w$ \emph{extends} $v$. If $F =
\langle v_1, \ldots, v_R \rangle \in \kvotep q A$ and $F' = \langle
v_1', \ldots, v_R'\rangle \in \kvotep q B$ with $v_i \leq v_i'$ for
every $1 \leq i \leq |R|$, we say that $F \leq F'$ and that $F$
extends $F'$.

An $A$ profile $F \in \kvotep{q}{A}$ is a \emph{complete} profile iff
the sum of its components equal $|\agents|$, \ie $F \in \votep{q}$ iff
$\left(\sum_{r \leq |R|}\sum_{a \in \act(q,r)}v(a)\right) = |\agents|$
iff $A = \agents$. Given a (partial) profile $F'$ at a state $q$ we
write $ext(q,F)$ for the set of all complete profiles that extend
$F'$.

Given two states $q,q' \in Q$, we say that $q'$ is a \emph{successor}
of $q$ if there is some $F \in \votep q$ such that $\delta(q,F) =
q'$. A \emph{computation} is an infinite sequence $\lambda =
q_0q_1\ldots$ of states such that for all positions $i\geq 0$,
$q_{i+1}$ is a successor of $q_i$. We follow the standard
abbreviations, hence $q$-computation denotes a computation starting at
$q$, and $\lambda[i]$, $\lambda[0,i]$ and $\lambda[i,\infty]$ denote
the $i$-th state, the finite prefix $q_0q_1\ldots q_i$ and the
infinite suffix $q_iq_{i+1}\ldots$ of $\lambda$ for any computation
$\lambda$ and its position $i\geq 0$. An \emph{$A$-strategy} for $A
\subseteq \agents$ is a function $s_A: Q \to \bigcup_{q \in Q}\kvotep
q A$ such that $s_A(q) \in \kvotep q A$ for all $q \in Q$. That is,
$s_A$ maps states to $A$-profiles at that state. The set of all
$A$-strategies is denoted by $\strats A$. If $s$ is an
$\agents$-strategy and we apply $\delta_{q}$ to $s(q)$, we obtain a
unique new state $q' = \delta_q(\strat(q))$. Iterating, we get the
\emph{induced} computation $\lambda_{s,q} = q_{0}q_{1}\ldots$ such
that $q = q_0$ and $\forall i~\geq 0: \delta_{q_{i}} (s(q_i)) =
q_{i+1}$. Given two strategies $s$ and $s'$, we say that $s \leq s'$
iff $\forall q \in Q: s(q) \leq s'(q)$. Given an $A$-strategy $s_A$
and a state $q$ we get an associated \emph{set} of computations
$out(s_A,q)$. This is the set of all computations that can result when
at any state, the players in $A$ are voting in
the way specified by $s_A$: \\
$out(s_A,q) = \{\lstrat {\strat,q} \mid \strat \textit{ is an}\
\mathcal{A}\textit{-strategy and } \strat \geq s_A\} $ It will also be
useful to have access to the set of states that can result in the next
step when $A \subseteq \agents$ follows strategy $s_A$ at state $q$,
$succ(q,s_A) = \{q' \in Q \mid \exists F \in ext(q,s_A): \delta(q,F) =
q'\}$. Clearly, $q' \in succ(q,s_A)$ iff there is some $\lambda \in
out(q,s_A)$ such that $q' = \lambda[0]$.

\subsection{New Semantics for ATL}

\begin{definition}\label{def:basicsem} 
  Given a \acro{rcgs} $S$ and a state $q$ in $S$, we define the
  satisfaction relation $\models$ inductively:
  \begin{align*}
  S,q \models &p \iff p \in \pi(q)\\
  S,q \models &\neg \phi \iff \text{ not } S,q \models \phi \\
  S,q \models &\phi \wedge \phi' \iff S,q \models \phi \text{ and } S,q \models \phi'\\
  S,q \models &\catld{A}\bigcirc\phi \iff \text{ there is } s_A \in \strats A \text{ such that } \\
 & \text{ for all } \lambda \in out(s_A,q) \text{, we have } S, \lambda[1] \models \phi \\
  S,q \models &\catld{A}\Box\phi \iff \text{ there is } s_A \in \strats A \text{ such that } \\
 & \text{for all } \lambda \in out(s_A,q) \text{ we have } S, \lambda[i] \models \phi \text{ for all } i \geq 0 \\
 S,q \models &\catld{A}\phi\mathcal{U}\phi' \iff \text{ there is }s_A \in \strats A \text{ such that } \\
 & \text{for all } \lambda \in out(s_A,q) \text{ we have } S, \lambda[i] \models \phi' \text{ and }S, \lambda[j] \models \phi\\
 &  \text{for some } i \geq 0 \text{ and for all }0 \leq j < i \\
\end{align*}
\end{definition}

\section{Equivalence between RCGS and CGS}
\label{sec:equiv}
In this section we show that definition \ref{def:basicsem} provides an
equivalent semantics for \acro{atl}. We do this by first giving a
surjective function $f$ that takes an \acro{rcgs} and returns a
\acro{cgs}. Then we show that $S$ and $f(S)$ satisfy the same
\acro{atl} formulas.

Remember that a~concurrent game structure is a~tuple $\langle
\agents,Q,\Pi,\pi,d,\delta' \rangle$ where every element is defined as
for an \acro{rcgs} except $d:\agents \times Q \to \Nat^{+}$ that maps
agents and states to actions available at that state, and $\delta'$ that is a partial function from states and action tuples to states defined by $\delta'(q,t) = \delta'_q(t)$ where $\delta'_q: \prod_{a \in \agents} [d_{a}(q)] \to Q$ is a transition function at $q$ based on tuples of actions rather than profiles. The
satisfaction relation for \acro{atl} based on \acro{cgs}s can be
defined exactly as in definition \ref{def:basicsem}, the difference
concerning only what counts as a strategy.

We refer to elements of $\prod_{a \in \agents} [d_{a}(q)]$ as
\emph{complete} action tuples at $q$.  A (memory-less) strategy for $a
\in \agents$ in a \acro{cgs} $M$ is a function $s_a: Q \to \mathbb
N^+$ such that for all $q \in Q$, $s_a(q) \in [d_a(q)]$ while a
strategy for $A \subseteq \agents$ is a list of strategies for all
agents in $A$, $s_A = \langle
s_{a_1},s_{a_2},\ldots,s_{a_{|A|}}\rangle$, for $A = \{a_1, a_2,
\dots, a_{|A|}\}$. We denote the set of strategies for $A \subseteq
\agents$ by $\strats A$. When needed to distinguish between different
structures we write $\strats {S,A}$ to indicate that we are talking
about the set of strategies for $A$ in $S$.

We say that a~complete action tuple at $q$, $t = \langle
i_{a_1},\ldots, i_{a_n} \rangle$ extends a strategy $s_A \in \strats
A$ if for all $a_j \in A$ we have $i_{a_j} = s_{a_j}(q)$. We denote
the set of all complete action tuples at $q$ extending $s_A$ by
$ext(q,s_A)$. For any state $q\in Q$ we have the set of all
computations that comply with $s_A$:

\begin{align*}
  out(q,s_A) = &\{\lambda = q_0q_1q_2\ldots \\
  &\mid q = q_0 \text{ and for all
  } i \in \Nat: \exists t \in ext(q_i,s_A) ,\ \delta(q,t) = q_{i+1}\}\\
\end{align*}
We define the set of $s_A$-successors at $q \in Q$:

$$ 
succ(q,s_A) = \{q' \in Q \mid \exists t \in ext(q,s_A) ,\ \delta(q,t)
= q'\}
$$
When we need to make clear which structure we are talking about, we
write $succ(S, q, s_A)$. Observe that $q' \in succ(q, s_A)$ iff $q' =
\lambda[1]$ for some $\lambda \in out(q, s_A)$.

The translation function $f$ from \acro{rcgs} to \acro{cgs} is defined
as follows:

$$ f \langle \agents, R, \roles, Q, \Pi, \pi, \act, \delta\rangle = \langle \agents, Q, \Pi, \pi, d, \delta' \rangle  $$
where:
\begin{align*}
  d_a(q) &= \act(q, r) & \text{where } a \in \roles(q, r) \\
  \delta'(q, \alpha_1, \ldots, \alpha_n) &= \delta(q, v_1, \dots, v_{|R|}) &\text{where for each role $r$} \\
\end{align*}
\vspace{-3.28em}
\[
v_r = \langle |\{i \in \roles(q, r)~|~ \alpha_{i} = 1\}|, \ldots ,
|\{i \in \roles(q, r)~|~ \alpha_{i} = \act(q,r)\}| \rangle
\]
% \smallskip %optional

We can see straight away that $f$ is surjective because for any
\acro{cgs} $S'$ with $n$ agents we could define a \acro{rcgs} $S$ with
that many roles where each role contains exactly one agent. A vote for
a role $r$, $v_r$, at $q$ would then simply be a $d_a(q)$-tuple
consisting of a single 1 (representing the agents chosen action) and
otherwise zeros. It is easy to verify that $f(S) = S'$.

Given either a~\acro{cgs} or an \acro{rcgs} $S$, we define the set of
sets of states that a~coalition $A$ can \emph{enforce} in the next
state of the game:
$$force(S,q,A) = \{succ(q,s_A) \mid s_A \text{ is a strategy for } A
\text{ in } S\}.$$

The first thing we do towards showing equivalence is to describe a
surjective function $m : \strats {f(S)} \to \strats S$ mapping action
tuples and strategies of $f(S)$ to profiles and strategies of $S$
respectively. For all $A \subseteq \agents$ and any action tuple for
$A$ at $q$, $t_q = \<\alpha_{a_1},\alpha_{a_2},...,\alpha_{a_{|A|}}\>$
with $1 \leq \alpha_{a_i} \leq d_{a_i}(q)$ for all $1 \leq i \leq
|A|$, the $A$-profile $m(t_q)$ is defined in the following way:

\begin{align*}
  m(t_q) &= \<v(t_q,1),\ldots,v(t_q,|R|)\> \text{ where for all } 1 \leq r \leq |R| \text{ we have } \\
  v(t_q,r) &= \<|\{a \in A_{r,q} \mid \alpha_a = 1\}|,\ldots,|\{a \in
  A_{r,q} \mid \alpha_a =
  \act(q,r)\}|\> \\
\end{align*}

Thus the $i$-th component of $v(t_q,r)$ will be the number of agents
from $A$ in role $r$ at $q$ that perform action $i$.

Given a strategy $s_A$ in $f(S)$ we define the strategy $m(s_A)$ for
$S$ by taking $m(s_A)(q) = m(s_A(q))$ for all $q \in Q$.

Surjectivity of $m$ is helpful since it means that for every possible
strategy that exists in the \acro{rcgs} $S$, there is a corresponding
one in $f(S)$. This in turn means that when we quantify over
strategies in one of $S$ and $f(S)$ we are implicitly also quantifying
over strategies in the other. Showing equivalence, then, can be done
by showing that these corresponding strategies have the same
strength. Before we proceed, we give a proof of surjectivity of $m$.

\begin{lemma}\label{lemma:surj}
  For any \acro{rcgs} $S$ and any $A \subseteq \agents$, the function
  $m: \strats{f(S),A} \to \strats{S,A}$ is surjective
\end{lemma}

\begin{proof}
  Let $p_A$ be some strategy for $A$ in $S$. We must show there is a
  strategy $s_A$ in $f(S)$ such that $m(s_A) = p_A$.  For all $q \in
  Q$, we must define $s_A(q)$ appropriately. Consider the profile
  $p_A(q) = \<v_1,\ldots,v_{|R|}\>$ and note that by definition of a
  profile, all $v_r$ for $1 \leq r \leq |R|$ are $A$-votes for $r$ and
  that by definition of an $A$-vote, we have $\sum_{1 \leq i \leq
    \act(q,r)}v_r(i) = |A_{r,q}|$. Also, for all agents $a,a' \in
  A_{r,q}$ we know, by definition of $f$, that $d_a(q) = d_{a'}(q) =
  \act(q,r)$.

  From this it follows that there are functions $\alpha: A \to \mathbb
  N^+$ such that for all $a \in A$, $\alpha(a) \in [d_a(q)]$ and $|\{a
  \in A_{r,q} \mid \alpha(a) = i\}| = v_r(i)$ for all $1 \leq i \leq
  \act(q,r)$, i.e.
$$v_r = \< |\{a \in A_{r,q} | \alpha (a) = 1\}|, \dots, 
|\{a \in A_{r,q} | \alpha (a) = \act(q,r)\}| \>$$ We choose some such
$\alpha$ and $s_A = \<\alpha(a_1), \dots, \alpha(a_{|A|})\>$. Having
defined $s_A$ in this way, it is clear that $m(s_A) = p_A$.
\end{proof}

Using the surjective function $m$ we can prove the following lemma,
showing that the "next time" strength of any coalition $A$ is the same
in $S$ as it is in $f(S)$.

\begin{lemma}\label{lemma:basic}
  For any \acro{rcgs} $S$, any state $q \in Q$ and any coalition $A
  \subseteq \agents$, we have $force(S,A,q) = force(f(S),A,q)$
\end{lemma}

\begin{proof}
  By definition of $force$ and lemma \ref{lemma:surj} it is sufficient
  to show that for all $s_A \in \strats{f(S),A}$, we have
  $succ(S,m(s_A),q) = succ(f(S),s_A,q)$. We show $\subseteq$ as
  follows: Assume that $q' \in force(S,m(s_A),q)$. Then there is some
  complete profile $P = \<v_1,\ldots,v_{|R|}\>$, extending
  $m(s_A)(q)$, such that $\delta(q,P) = q'$. Let $m(s_A)(q) =
  \<w_1,\ldots,w_{|R|}\>$ and form $P' = \<v'_1,\ldots,v'_{|R|}\>$
  defined by $v'_i = v_i - w_i$ for all $1 \leq i \leq |R|$. Then each
  $v'_i$ is an $(\agents \setminus A)$-vote for role $i$, meaning that
  the sum of entries in the tuple $v'_i$ is $|(\agents \setminus
  A)_{r,q}|$.  This means that we can define a function $\alpha:
  \agents \to \Nat^+$ such that for all $a \in \agents$, $\alpha(a)
  \in [d_a(q)]$ and for all $a \in A$, $\alpha(a) = s_a(q)$ and for
  every $r \in R$ and every $a \in (\agents \setminus A)$, and every
  $1 \leq j \leq \act(q,r)$, $|\{a \in (\agents \setminus A)_{r,q}
  \mid \alpha(a) = j\}| = v'_r(j)$. Having defined $\alpha$ like this
  it follows by definition of $m$ that for all $1 \leq j \leq
  \act(q,r)$, $|\{a \in A_{r,q} \mid \alpha(a) = j\}| = w_r(j)$. Then
  for all $r \in R$ and all $1 \leq j \leq \act(q,r)$ we have $|\{a
  \in \roles(q,r) \mid \alpha(a) = j\}| = v_r(j)$. By definition of
  $f(S)$ it follows that $q' = \delta(q,P) = \delta'(q,\alpha)$ so
  that $q' \in force(f(S),s_A,q)$. We conclude that $force(S,f(s_A),q)
  \subseteq force(f(S),s_A,q)$. The direction $\supseteq$ follows
  easily from the definitions of $m$ and $f$.
\end{proof}

Given a structure $S$ (with or without roles), and a formula $\phi$,
we define $true(S,\phi) = \{q \in Q \mid S,q \models
\phi\}$. Equivalence of models $S$ and $f(S)$ is now demonstrated by
showing that the equivalence in next time strength established in
lemma \ref{lemma:basic} suffices to conclude that $true(S,\phi) =
true(f(S),\phi)$ for all $\phi$.

\begin{theorem}\label{thm:basicthm}
  For any \acro{rcgs} $S$, any $\phi$ and any $q \in Q$, we have $S,q
  \models \phi$ iff $f(S),q \models_{CGS} \phi$
\end{theorem}

\begin{proof}
  We prove the theorem by showing that for all $\phi$, we have
  $true(S,\phi) = true(f(S),\phi)$. We use induction on complexity of
  $\phi$. The base case for atomic formulas and the inductive steps
  for Boolean connectives are trivial, while the case of
  $\catld{A}\bigcirc\phi$ is a straightforward application of lemma
  \ref{lemma:basic}. For the cases of $\catld{A}\Box\phi$ and
  $\catld{A}\phi\mathcal{U}\psi$ we rely on the following fixed point
  characterizations, which are well-known to hold for \acro{atl}, see
  for instance \cite{Jamroga:2009:EYH:1615285.1615287}, and are also easily
  verified against definition \ref{def:basicsem}:
  \begin{gather}\label{eq:fixed}
    \begin{gathered}
      \catld{A}\Box\phi \leftrightarrow \phi \land
      \catld{A}\bigcirc\catld{A}\Box\phi\\
      \catld{A}\phi_{1}\mathcal{U}\phi_{2} \leftrightarrow
      \phi_{2}\lor
      (\phi_{1}\land\catld{A}\bigcirc\catld{A}\phi_{1}\mathcal{U}\phi_{2}
    \end{gathered}
  \end{gather}
  We show the induction step for $\catld{A}\Box\phi$, taking as
  induction hypothesis $true(S,\phi) = true(f(S),\phi)$.  The first
  equivalence above identifies $Q' = true(S,\catld{A}\Box\phi)$ as the
  maximal subset of $Q$ such that $\phi$ is true at every state in
  $Q'$ and such that $A$ can enforce a state in $Q'$ from every state
  in $Q'$, i.e. such that $\forall q \in Q': \exists Q'' \in
  force(q,A): Q'' \subseteq Q'$.  Notice that a unique such set always
  exists. This is clear since the union of two sets satisfying the two
  requirements will itself satisfy them (possibly the empty set). The
  first requirement, namely that $\phi$ is true at all states in $Q'$,
  holds for $S$ iff if holds for $f(S)$ by induction hypothesis. Lemma
  \ref{lemma:basic} states $force(S,q,A) = force(f(S),q,A)$, and this
  implies that also the second requirement holds in $S$ iff it holds
  in $f(S)$. From this we conclude $true(S,\catld{A}\Box\phi) =
  true(f(S),\catld{A}\Box\phi)$ as desired. The case for
  $\catld{A}\phi\mathcal{U}\psi$ is similar, using the second
  equivalence.\qed
\end{proof}

\section{Model checking and the size of models}
\label{sec:mc}

We have already seen that using roles can lead to a dramatic decrease
in the size of \acro{atl}-models. In this section we give a more
formal account, first by investigating the size of models in terms of
the number of roles, players and actions, then by an analysis of model
checking \acro{atl} over concurrent game structures with roles.

Given a set of numbers $[a]$ and a number $n$, it is a well-known
combinatorial fact that the number of ways in which to choose $n$
elements from $[a]$, allowing repetitions, is $\frac{(n +
  (a-1))!}{n!(a-1)!}$. Furthermore, this number satisfies the
following two inequalities:\footnote{If this is not clear, remember
  that $n^a$ and $a^n$ are the number of functions $[n]^{[a]}$ and
  $[a]^{[n]}$ respectively. It should not be hard to see that all ways
  in which to choose $n$ elements from $a$ induce non-intersecting
  sets of functions of both types}
% Clearly, every way in which to choose $n$ elements from $a$ with repetition corresponds to a unique class of functions, namely those $f: [a] \to [n]$ for which for all $i \in [a]$ we have that $f(i)$ is the number of times $i$ was chosen if it was chosen. Some functions do not meet this requirement for any manner in which to choose $n$ elements from $[a]$, so the second inequality holds and is strict. To see that the first holds, note that any way in which to choose $n$ elements from $[a]$ corresponds to a class of functions $[n] \to [a]$, namely those $g: [n] \to [a]$ for which it is true that if $i \in [a]$ was chosen $j$ times then there are $j$ distinct elements $k \in [n]$ such that $f(k) = i$. It clear that not every function $[n] \to [a]$ has this property, so this equality is strict also.}\\
\begin{equation}\label{eq:basic}
  \begin{array}{ccc}\frac{(n + (a-1))!}{n!(a-1)!} \leq a^n &, & \frac{(n + (a-1))!}{n!(a-1)!} \leq n^a\end{array}
\end{equation}

These two inequalities provide us with an upper bound on the
\emph{size} of \acro{rcgs} models that makes it easy to compare their
sizes to that of \acro{cgs} models.  Typically, the size of concurrent
game structures is dominated by the size of the domain of the
transition function. For an \acro{rcgs} and a given state $q \in Q$
this is the number of complete profiles at $q$. To measure it,
remember that every complete profile is an $|R|$-tuple of votes $v_r$,
one for each role $r \in R$. It follows that $|\votep q|$ is the set
of all possible combinations of votes for each role. Also remember
that a vote $v_r$ for $r \in R$ is an $\act(q,r)$-tuple such that the
sum of entries is $|\roles(q,r)|$. Equivalently, the vote $v_r$ can be
seen as the number of ways in which we can make $|\roles(q,r)|$
choices, allowing repetitions, from a set of $\act(q,r)$
alternatives. Looking at it this way, we obtain:
$$|\votep q| = \prod_{r\in
  R}\frac{(|\roles(q,r)| + (\act(q,r) -
  1))!}{|\roles(q,r)|!(\act(q,r)-1))!}
$$ 
  
We sum over all $q \in Q$ to obtain what we consider to be the size of
an \acro{rcgs} $S$. In light of equation \ref{eq:basic}, it follows
that the size of $S$ is upper bounded by both of the following
expressions.

\begin{equation}\label{eq:size}
  \begin{array}{ccc} \bigO (\sum_{q\in Q}\prod_{r \in R}|\roles(q,r)|^{\act(q,r)}) & , & \bigO (\sum_{q \in Q}\prod_{r\in R}\act(q,r)^{|\roles(q,r)|})
  \end{array}
\end{equation}
We observe that growth in the size of models is polynomial in $a =
max_{q \in Q,r \in R}\act(r,q)$ if $n = \agents$ and $|R|$ is fixed,
and polynomial in $p = max_{q \in Q, r \in R}|\roles(q,r)|$ if $a$ and
$|R|$ are fixed. This identifies a significant potential advantage
arising from introducing roles to the semantics of \acro{atl}. The
size of a \acro{cgs} $M$, when measured in the same way, replacing
complete profiles at $q$ by complete action tuples at $q$, grows
exponentially in the players whenever $d_a(q) > 1$ for each player
$a$. We stress that we are \emph{not} just counting the number of
transitions in our models differently. We do have an additional
parameter, the roles, but this is a genuinely new semantic construct
that gives rise to genuinely different semantic structures. We show
that it is possible to use them to give the semantics of \acro{atl},
but this does not mean that there is not more to be said about
them. Particularly crucial is the question of model checking over
\acro{rcgs} models.

\subsection{Model checking using roles}

For strategic logics, checking satisfiability is usually
non-tractable, and the question of model checking is often crucial in
assessing the usefulness of different logics. For \acro{atl} there is
a well known ``standard'' algorithm, see
e.g. \cite{alur2002alternating}. It does model checking in time linear
in the length of the formula and the size of the model. The algorithm
is based on the fixed point equation \ref{eq:fixed} from the proof of
Theorem \ref{thm:basicthm}, so it will work also when model checking
\acro{rcgs} models. It is not clear, however, how the high level
description should be implemented and, crucially, what the complexity
will be in terms of the new parameters that arise.

Given a structure with roles, $S$, and a formula $\phi$, the standard
model checking algorithm returns the set $true(S,\phi)$, proceeding as
detailed in algorithms \ref{alg:pre} and \ref{alg:force}.

\begin{algorithm}[h]\refstepcounter{theorem}
  \caption{$mcheck(S,\phi)$}
  \label{alg:pre}
  \begin{algorithmic}
    \IF{$\phi = p \in \Pi$} \RETURN $\pi(p)$
    \ENDIF
    \IF {$\phi = \neg \psi$} \RETURN $Q \setminus mcheck(S,\psi)$
    \ENDIF
    \IF {$\phi = \psi \wedge \psi'$} \RETURN $mcheck(S,\psi) \cap
    mcheck(S,\psi')$
    \ENDIF
    \IF {$\phi = \catld{A}\bigcirc\psi$} \RETURN $\{q \mid
    enforce(S,q,A, mcheck(S,\psi))\}$
    \ENDIF
    \IF {$\phi = \catld{A}\Box\psi$} \STATE $Q_1 := Q$, $Q_2 :=
    mcheck(S,\psi)$ \WHILE {$Q_1 \not \subseteq Q_2$} \STATE $Q_1 :=
    Q_2$, $Q_2 := \{q \in Q \mid enforce(S,A,q,Q_2)\} \cap Q_2$
    \RETURN $Q_1$
    \ENDWHILE
    \ENDIF
    \IF {$\phi = \catld{A}\psi\mathcal{U}\psi'$} \STATE $Q_1 :=
    \emptyset$, $Q_2 = mcheck(S,\psi)$, $Q_3 = mcheck(S,\psi')$ \WHILE
    {$Q_3 \not \subseteq Q_1$} \STATE $Q_1 := Q_1 \cup Q_3$, $Q_3 :=
    \{q \in Q \mid enforce(S,A,q,Q_1)\} \cap Q_2$ \RETURN $Q_3$
    \ENDWHILE
    \ENDIF
  \end{algorithmic}
\end{algorithm}

Given a structure $S$, a coalition $A$, a state $q \in Q$ and a set of
states $Q'$, the method $enforce$ answers true or false depending on
whether or not $A$ can enforce $Q'$ from $q$. That is, it tells us if
at $q$ there is $Q'' \in force(q,A)$ such that $Q'' \subseteq
Q'$. Given a fixed length formula and a fixed number of states, this
step dominates the running time of $mcheck$ (algorithm
\ref{alg:pre}). It is also the only part of the standard algorithm
that behaves in a different way after addition of roles to the
structures. It involves the following steps:\footnote{In
  implementations one would seek to take advantage of information
  collected by repeating calls to $enforce$ and not just do a Boolean
  check for every new instance in the way we do it here.  This aspect
  is not crucial for our analysis, so we do not address it further}

\begin{algorithm}[h]\refstepcounter{theorem}
  \caption{$enforce(S,A,q,Q')$}
  \label{alg:force}
  \begin{algorithmic}
    \FOR {$F \in \kvotep q A$} \STATE $p = true$ \FOR {$F' \in
      ext(q,F)$} \IF {$\delta(q,F') \not \in Q'$} \STATE $p = false$
    \ENDIF
    \ENDFOR
    \IF {$p = true$} \RETURN $true$
    \ENDIF
    \ENDFOR
    \RETURN $false$
  \end{algorithmic}
\end{algorithm}

For all profiles $F \in \kvotep q A$ the algorithm runs through all
complete profiles $F' \in \votep q$ that extend $F$. Over \acro{cgs}s,
given a coalition $A$ and two action tuples $t =
\<\alpha_{a_1},\alpha_{a_2},\ldots,\alpha_{a_{|A|}}\>, t' =
\<\alpha'_{a_1},\alpha'_{a_2},\ldots,\alpha'_{a_{|A|}}\>$ for $A$ at
$q$, the sets of complete action tuples that extend $t$ and $t'$
respectively do not intersect. It follows that running through all
such extensions for all possible action tuples for $A$ at $q$ is at
most linear in the total number of complete action tuples at $q$. This
is no longer the case for \acro{rcgs} models. Given two profiles
$P,P'$ for $A$ at $q$, there can be many shared extensions. In fact,
$P$ and $P'$ can share exponentially many in terms of the number of
players and actions available.\footnote{ To see this, consider $P =
  \<v_1,v_2\ldots,v_{|R|}\>$ and $P' =
  \<v'_1,v'_2,\ldots,v'_{|R|}\>$. Each $v_r,v'_r \in \kvote A q r$
  sums to $\Sigma_{1 \leq j \leq \act(q,r)} v_i(j) = |A_{q,r}|$. Then
  form a complete profile $P'' = \<v''_1,v''_2,\ldots,v''_{|R|}\>$ at
  $q$ such that for all $1 \leq r \leq |R|$ and all $1 \leq j \leq
  \act(q,r)$ we have $v''_r(j) = max(v_r(j),v'_r(j))$. Then, if it
  exists, choose a coalition $A'$ such that $|A'_{r,q}| = \Sigma_{1
    \leq j \leq \act(q,r)}v''_r(j)$. It is clear that the number of
  complete profiles that extends both $v$ and $v'$ is equal to the
  number of \emph{all} $\agents \setminus A'$-profiles at $q$.} So, in
general, running $enforce$ requires us to make several passes through
the set of all complete profiles, and the complexity is no longer
linear. Still, it is polynomial in the number of complete profiles,
since for any coalition $A$ and state $q$ we have $|\kvotep q A| \leq
|\votep q|$, meaning that the complexity of $enforce$ is upper bounded
by $|\votep q|^2$. It follows that model checking of \acro{atl} over
concurrent game structures with roles is polynomial in the size of the
model. We summarize this result.

\begin{proposition}\label{prop:modelcheck}

  Given a \acro{cgs} $S$ and a formula $\phi$, $mcheck(S,\phi)$ takes
  time $\bigO(l e^2)$ where $l$ is the length of $\phi$ and $e =
  \sum\limits_{q \in Q}\votep q$ is the total number of transitions in
  $S$
\end{proposition}

Since model checking \acro{atl} over \acro{cgs}s takes only linear
time, $\bigO (l e)$, adding roles apparently makes model checking
harder. On the other hand, the \emph{size} of \acro{cgs} models can be
bigger by an exponential factor, making model checking much easier
after adding roles. In light of the bounds we have on the size of
models, c.f. equation \ref{eq:size}, we find that as long as the roles
and the actions remain fixed, complexity of model checking is only
polynomial in the number of agents. This is a potentially significant
argument in favor of roles.

In practice, however, finding an optimal \acro{rcgs} for a given
\acro{cgs} model $M$ might be at least as difficult as model checking
on $M$ directly. It involves identifying the structure from $f^-(M)$
that has the minimum number of roles. In general, one cannot expect
this task to have sub-linear complexity in the size of
$M$.\footnote{Although in many practical cases, when models are given
  in some compressed form, the situation might be such that it is
  possible. The question of how to efficiently find small
  \acro{rcgs}-models will be investigated in future work.} Roles
should be used at the modelling stage, as they give the modeller an
opportunity for exploiting homogeneity in the system under
consideration. We think that it is reasonable to hypothesize that in
practice, most large scale systems that lends themselves well to
modelling by \acro{atl} do so precisely because they exhibit
significant homogeneity. If not, identifying an accurate \acro{atl}
model of the system, and model checking it, seems unlikely to be
tractable at all.

The question arises as to whether or not using an \acro{rcgs} is
\emph{always} the best choice, or if there are situations when the
losses incurred in the complexity of model checking outweigh the gains
we make in terms of the size of models. A general investigation of
this in terms of how fixing or bounding the number of roles affect
membership in complexity classes is left for future work. Here, we
conclude with the following proposition which states that as long we
use the standard algorithm, model checking any \acro{cgs} $M$ can be
done at least as quickly by model checking an \emph{arbitrary} $S \in
f^-(M)$.

\begin{proposition}\label{prop:alwaysbest}
  Given any \acro{cgs}-model $M$ and any formula $\phi$, let
  $c(mcheck(M,\phi))$ denote the complexity of running
  $mcheck(M,\phi)$. We have, for all $S \in f^-(M)$, that complexity
  of running $mcheck(S,\phi)$ is $\bigO (c(mcheck(M,\phi))$
\end{proposition}

\begin{proof}
  It is clear that for any $S \in f^-(M)$, running $mcheck(S,\phi)$
  and $mcheck(M,\phi)$, a difference in overall complexity can arise
  only from a difference in the complexity of $enforce$. So we compare
  the complexity of $enforce(S,A,q,Q'')$ and $enforce(M,A,q,Q'')$ for
  some arbitrary $q \in Q$, $Q'' \subseteq Q$. The complexity in both
  cases involves passing through all complete extensions of all
  strategies for $A$ at $q$. The sizes of these sets are can be
  compared as follows, the first inequality is an instance of equation
  \ref{eq:basic} and the equalities follow from definition of $f$ and
  the fact that $M = f(S)$.
  \begin{align*}
    \prod_{r \in R} \left(\frac{(|A_{r,q}| + (\act(r,q)- 1))!}{|A_{r,q}|!(\act(r,q)-1)!}\right) \times \prod_{r \in R}\left(\frac{((|\roles(r,q)| - |A_{r,q}|) + (\act(r,q) - 1))!}{(|\roles(r,q)| - |A_{r,q}|)!(\act(r,q)-1)!}\right) & \\ \leq \left(\prod_{r \in R} \act(r,q)^{|A_{r,q}|} \times \prod_{r \in R} \act(r,q)^{|\roles(r,q)| - |A_{r,q}|}\right) \\
    = \prod_{r \in R} \left(\prod_{a \in A_{r,q}} \act(r,q)\right) \times \prod_{r\in R} \left(\prod_{a \in \roles(a,r) \setminus A_{r,q}} \act(r,q)\right) \\
    = \left(\prod_{a \in A}d_a(q) \times \prod_{a \in \agents
        \setminus
        A}d_a(q)\right) = \prod_{a \in \agents}d_a(q)\\
  \end{align*}

  We started with the number of profiles (transitions) we need to
  inspect when running $enforce$ on $S$ at $q$, and ended with the
  number of action tuples (transitions) we need to inspect when
  running $enforce$ on $M = f(S)$. Since we showed the first to be
  smaller or equal to the latter and the execution of all other
  elements of $mcheck$ are identical between $S$ and $M$, the claim
  follows.
\end{proof}

\section{Conclusions, related and future work}
\label{sec:concl}

In this paper we have described a~new type of semantics for the
strategic logic \acro{atl}. We have provided motivational examples and
argued that although in principle model checking \acro{atl}
interpreted over concurrent game structures with roles is harder than
the standard approach, it is still polynomial and generates
exponentially smaller models. We believe this provides conclusive
evidence that concurrent game structures with roles are an interesting
semantics for \acro{atl}, and should be investigated further.

Relating our work to ideas already present in the literature we find
it somewhat similar to the concept of exploiting symmetry in model
checking, as investigated by Sistla and Godefroid
\cite{SisGod0407-0}. Our approach is however different, since we we
only look at agent symmetries in \acro{atl}. When it comes to work
related directly to strategic logics, we find no similar ideas
present, hence concluding that our approach is indeed novel.

For future work we plan on investigating the homogeneous aspect of our
`roles' in more depth. We are currently working on a~derivative of
\acro{atl} with a~different language that will fully exploit the role
based semantics.

\paragraph{Acknowledgements} We thank anonymous reviewers of LAMAS
2012 and Pål Grønås Drange for helpful comments.

\bibliographystyle{plain} \bibliography{references.bib}

\end{document}